\documentclass[10pt, conference, letterpaper]{IEEEtran}

\usepackage{amsmath,amssymb,amsfonts}
\usepackage{algorithmic}
\usepackage{graphicx}
\usepackage{textcomp}
\usepackage{booktabs}
\usepackage{multirow}
\usepackage{cite}
\usepackage{url}
\usepackage{xcolor}
\usepackage{enumitem}
\usepackage{adjustbox}
\usepackage{tikz}
\usetikzlibrary{positioning, arrows.meta, calc}

\newcommand{\decision}[1]{\textsf{\MakeUppercase{#1}}}

\usepackage{amsthm}
\newtheorem{definition}{Definition}
\newtheorem{theorem}{Theorem}

\begin{document}

\title{Validated Intent Compilation for Constrained Routing\\in LEO Mega-Constellations}

\author{\IEEEauthorblockN{Yuanhang Li}}

\maketitle

\begin{abstract}
Operating LEO mega-constellations requires translating high-level
operator intents (``reroute financial traffic away from polar links
under 80\,ms'') into low-level routing constraints---a task that
demands both natural language understanding and network-domain
expertise. We present an end-to-end system comprising three components:
(1)~a GNN cost-to-go router that distills Dijkstra-quality routing
into a 152K-parameter graph attention network achieving 99.8\% packet
delivery ratio with 17$\times$ inference speedup;
(2)~an LLM intent compiler that converts natural language to a typed
constraint intermediate representation using few-shot prompting with
a verifier-feedback repair loop, achieving 98.4\% compilation rate
and 87.6\% full semantic match on feasible intents in a 240-intent
benchmark (193 feasible, 47 infeasible); and
(3)~an 8-pass deterministic validator with constructive feasibility
certification that achieves 0\% unsafe acceptance on all 47 infeasible intents
(30 labeled + 17 discovered by Pass~8), with 100\% corruption
detection across 240 structural corruption tests and 100\% on 15
targeted adversarial attacks. End-to-end evaluation across four constrained
routing scenarios confirms zero constraint violations with both
routers. We further demonstrate that apparent performance gaps in
polar-avoidance scenarios are largely explained by topological
reachability ceilings rather than routing quality, and that the LLM
compiler outperforms a rule-based baseline by 46.2 percentage points
on compositional intents. Our system bridges the semantic gap between
operator intent and network configuration while maintaining the
safety guarantees required for operational deployment.
\end{abstract}


\section{Introduction}
\label{sec:intro}

Low Earth Orbit (LEO) mega-constellations such as Starlink, OneWeb,
and Kuiper are transforming global connectivity by deploying thousands
of satellites interconnected via inter-satellite links (ISLs). Operating
these networks presents unique challenges: the topology changes
continuously as satellites orbit, polar regions experience periodic
link dropout, and operators must enforce complex routing constraints
spanning latency guarantees, region avoidance, node maintenance, and
traffic prioritization.

Today, translating operator intent into network configuration requires
manual specification of routing policies---a process that is slow,
error-prone, and does not scale to the dynamic nature of LEO
constellations. Intent-based networking (IBN)~\cite{ibn_survey} promises
to bridge this gap by allowing operators to express high-level goals
that are automatically compiled into network configurations. However,
existing IBN approaches target terrestrial networks with relatively
stable topologies and do not address the unique constraints of satellite
mega-constellations.

We identify three key challenges in intent-driven LEO routing:

\begin{enumerate}[leftmargin=*]
\item \textbf{Compositional intent understanding.} Operator intents
combine multiple constraint types (``disable plane~7, avoid polar links
above 75\textdegree, and cap utilization at 80\%'') that must be
correctly decomposed and mapped to formal constraint representations.
Rule-based parsers handle simple intents but degrade sharply on
compositional ones (40\% vs.\ 86.2\% accuracy).

\item \textbf{Safety-critical verification.} In production networks,
a single undetected constraint violation can cascade into service
outages. The compiler's output must be formally verified before
reaching the routing layer, yet verification must be fast enough
for interactive use.

\item \textbf{Efficient constrained routing.} Applying constraints
modifies the network topology (disabling nodes, removing edges),
and the router must compute valid paths on the constrained graph
in real time. Traditional shortest-path algorithms are correct but
too slow for per-packet decisions at scale.
\end{enumerate}

We address these challenges with a three-component system:

\noindent\textbf{GNN Cost-to-Go Router} (Section~\ref{sec:gnn_router}).
A 3-layer graph attention network (152K parameters) trained via
supervised distillation from Dijkstra shortest paths. It achieves
99.8\% packet delivery ratio (PDR) while providing 17$\times$ inference
speedup, enabling real-time per-packet routing decisions.

\noindent\textbf{LLM Intent Compiler} (Section~\ref{sec:compiler}).
A Qwen3.5-9B language model with 6-shot prompting converts natural
language intents to a typed ConstraintProgram intermediate
representation. A verifier-feedback repair loop corrects compilation
errors, achieving 98.4\% compilation rate and 87.6\% full semantic
match on 193 feasible benchmark intents.

\noindent\textbf{Deterministic Validator} (Section~\ref{sec:validator}).
An 8-pass verification pipeline checks schema validity, entity
grounding, type safety, value ranges, constraint conflicts, physical
admissibility, and reachability. It achieves 100\% detection on
structural corruptions and guarantees that no malformed constraint
program reaches the routing layer.

Our contributions are:
\begin{itemize}[leftmargin=*]
\item A typed constraint IR (ConstraintProgram) that formally bridges
natural language intents and topology-level routing constraints, with
grounding semantics for 10 hard constraint types and support for soft
constraints with configurable penalty weights.
\item An LLM-based intent compiler with verifier-feedback repair that
outperforms rule-based parsing by 46.2pp on compositional intents and
generalizes to out-of-distribution phrasings (81.8\% accuracy, 4.4pp
degradation).
\item A reachability separation analysis showing that apparent routing
performance gaps under polar constraints are largely explained by
topological reachability ceilings, not routing quality.
\item End-to-end evaluation demonstrating zero constraint violations
across four scenarios with both GNN and Dijkstra routers.
\end{itemize}

\section{Related Work}
\label{sec:related}

\textbf{LEO constellation routing.}
Routing in LEO mega-constellations has been studied extensively.
Snapshot-based approaches~\cite{snapshot_routing} precompute routes
for discrete time intervals but cannot adapt to dynamic failures.
Contact graph routing (CGR)~\cite{cgr} handles time-varying topologies
but assumes deterministic contact schedules. Recent work applies
deep reinforcement learning (DRL) to LEO routing~\cite{drl_leo1,drl_leo2},
but DRL agents require online training and struggle with credit
assignment over large action spaces. Our GNN cost-to-go approach
avoids these issues through offline supervised distillation, achieving
Dijkstra-equivalent quality with 17$\times$ speedup.

\textbf{GNN-based network optimization.}
Graph neural networks have shown promise for combinatorial network
problems including traffic engineering~\cite{gnn_te}, link
scheduling~\cite{gnn_scheduling}, and routing~\cite{gnn_routing}.
RouteNet~\cite{routenet} models network performance but does not
produce routing decisions. Our work differs by training a GNN to
directly predict per-destination next-hop decisions via cost-to-go
distillation, enabling deployment as a drop-in routing engine.

\textbf{Intent-based networking.}
IBN aims to translate operator goals into network
configurations~\cite{ibn_survey,ibn_framework}. Existing systems
use template matching~\cite{ibn_template}, ontology-based
parsing~\cite{ibn_ontology}, or domain-specific languages~\cite{nile}.
Recent work explores LLMs for network intent
translation~\cite{llm_network1,llm_network2}, but without formal
verification of the compiled output. Our system combines LLM
compilation with deterministic verification, ensuring that the
semantic flexibility of LLMs does not compromise network safety.

\textbf{LLMs for network management.}
Large language models are increasingly applied to network
tasks~\cite{netllm,llm_netops}: configuration
generation~\cite{llm_config}, anomaly diagnosis~\cite{llm_anomaly},
and policy translation~\cite{llm_policy}. However, most approaches
trust LLM output directly or use only syntactic validation. Our
8-pass validator goes beyond syntax to check entity grounding, type
safety, physical admissibility, and reachability---providing the
verification depth required for safety-critical infrastructure.

\section{System Design}
\label{sec:system}


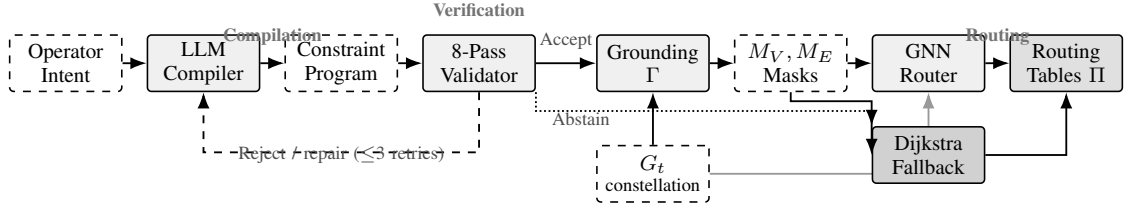
\begin{figure*}[t]
\centering
\begin{adjustbox}{max width=\textwidth}
\begin{tikzpicture}[
    node distance=0.45cm and 0.32cm,
    box/.style={draw=black, rounded corners=1.5pt, line width=0.6pt,
                minimum height=0.75cm, text width=1.35cm,
                inner xsep=2pt, inner ysep=2pt,
                align=center, font=\footnotesize},
    comp/.style={box, fill=black!6},
    data/.style={box, fill=white, dashed},
    result/.style={box, fill=black!12},
    fallback/.style={box, fill=black!18},
    arr/.style={->, >=Latex, line width=0.7pt},
    reject/.style={arr, dashed},
    abstain/.style={arr, densely dotted},
    lbl/.style={font=\scriptsize, text=black!70},
]

\node[data] (intent) {Operator\\Intent};
\node[comp, right=of intent] (compiler) {LLM\\Compiler};
\node[data, right=of compiler] (ir) {Constraint\\Program};
\node[comp, right=of ir] (validator) {8-Pass\\Validator};
\node[comp, right=0.8cm of validator] (grounding) {Grounding\\$\Gamma$};
\node[data, right=of grounding] (masks) {$M_V, M_E$\\Masks};
\node[comp, right=of masks] (router) {GNN\\Router};
\node[result, right=of router] (routes) {Routing\\Tables $\Pi$};

\draw[arr] (intent) -- (compiler);
\draw[arr] (compiler) -- (ir);
\draw[arr] (ir) -- (validator);

\draw[arr] (validator) -- (grounding);
\node[lbl, above=2pt] at ($(validator)!0.5!(grounding)$) {Accept};

\draw[arr] (grounding) -- (masks);
\draw[arr] (masks) -- (router);
\draw[arr] (router) -- (routes);

\draw[reject] (validator.south) -- ++(0,-0.8) -| (compiler.south);
\node[lbl] at ($(compiler.south)!0.5!(validator.south)+(0,-0.83)$)
  {Reject / repair ($\leq$3 retries)};

\node[data, below=0.7cm of grounding] (graph) {$G_t$\\[-1pt]\scriptsize constellation};
\draw[arr] (graph) -- (grounding);
\draw[arr, black!40] (graph) -| (router);

\node[fallback, below=0.45cm of router] (dijkstra) {Dijkstra\\Fallback};
\draw[arr] (masks.south) -- ++(0,-0.12) -| (dijkstra.west);
\draw[arr] (dijkstra) -| (routes.south);

\draw[abstain] (validator.south east) -- ++(0,-0.25) -| (dijkstra.north west);
\node[lbl] at ($(validator.south east)+(0.6,-0.35)$) {Abstain};

\node[above=0.12cm of $(compiler)!0.5!(ir)$, font=\scriptsize\bfseries, text=black!55] {Compilation};
\node[above=0.12cm of validator, font=\scriptsize\bfseries, text=black!55] {Verification};
\node[above=0.12cm of $(router)!0.5!(routes)$, font=\scriptsize\bfseries, text=black!55] {Routing};

\end{tikzpicture}
\end{adjustbox}
\caption{End-to-end system architecture. Operator intents are compiled
to typed ConstraintPrograms and verified by the 8-pass validator.
Accept (solid) proceeds to grounding and GNN routing;
Reject (dashed) triggers repair; Abstain (dotted) defers to Dijkstra.}
\label{fig:system_architecture}
\end{figure*}

\subsection{Problem Formulation}
\label{sec:formulation}

We consider a Walker Delta constellation with $P$ orbital planes and
$S$ satellites per plane, yielding $N = P \times S$ nodes. Each
satellite maintains up to $K=4$ ISL links (2 intra-plane, 2
inter-plane). The constellation graph $G_t = (V, E_t)$ evolves over
time as satellite positions change and polar links experience periodic
dropout above inclination $\iota$.

An operator intent $\ell$ is a natural language string expressing
routing constraints. The system must:
\begin{enumerate}
\item \textit{Compile}: $\ell \xrightarrow{\text{LLM}} \mathcal{P}$,
mapping intent to a ConstraintProgram $\mathcal{P}$.
\item \textit{Verify}: $\mathcal{P} \xrightarrow{\text{Validator}}
\mathcal{P}^{\checkmark}$, ensuring structural and physical validity.
\item \textit{Ground}: $(\mathcal{P}^{\checkmark}, G_t)
\xrightarrow{\Gamma} (M_V, M_E, \mathbf{u}, \mathbf{d})$, producing
topology masks and flow constraints.
\item \textit{Route}: $(G_t \odot (M_V, M_E), \mathbf{d})
\xrightarrow{\text{GNN}} \Pi$, computing constrained routing tables.
\end{enumerate}


\begin{definition}[ConstraintProgram]
A \textit{ConstraintProgram} $\mathcal{P}$ is a typed intermediate representation
that captures operator intent as a tuple:
\begin{equation}
\mathcal{P} = \langle \mathcal{F}, \mathcal{H}, \mathcal{S}, \mathcal{E}, \omega, \pi, \beta \rangle
\end{equation}
where $\mathcal{F}$ is a set of flow selectors, $\mathcal{H}$ is a set of hard constraints,
$\mathcal{S}$ is a set of soft constraints, $\mathcal{E}$ is a set of event conditions,
$\omega$ is an objective weight vector, $\pi \in \{\texttt{critical}, \texttt{high},
\texttt{medium}, \texttt{low}\}$ is the priority level, and $\beta$ is a fallback
policy governing behavior when hard constraints cannot be satisfied at routing time.
\end{definition}

\begin{definition}[Flow Selector]
A flow selector $f \in \mathcal{F}$ identifies a subset of traffic flows:
\begin{equation}
f = \langle \tau, r_s, r_d, n_s, n_d, p_s, p_d \rangle
\end{equation}
where $\tau \in \mathcal{T}$ is a traffic class (e.g., \texttt{financial}, \texttt{emergency}),
$r_s, r_d \in \mathcal{R} \cup \{\bot\}$ are source/destination regions,
$n_s, n_d \in [0, N) \cup \{\bot\}$ are source/destination node IDs,
and $p_s, p_d \in [0, P) \cup \{\bot\}$ are source/destination orbital planes.
\end{definition}

\begin{definition}[Hard Constraint]
A hard constraint $h \in \mathcal{H}$ must be satisfied; violation renders the
program infeasible:
\begin{equation}
h = \langle t_h, \sigma, v, c \rangle
\end{equation}
where $t_h \in \mathcal{T}_H$ is the constraint type, $\sigma$ is the target specifier,
$v$ is the constraint value, and $c \in \mathcal{E} \cup \{\bot\}$ is an optional
event condition.

The hard constraint type set is:
\begin{align}
\mathcal{T}_H = \{&\texttt{disable\_node}, \texttt{disable\_plane}, \nonumber \\
                   &\texttt{disable\_edge}, \texttt{avoid\_region}, \nonumber \\
                   &\texttt{avoid\_latitude}, \texttt{reroute\_away}, \nonumber \\
                   &\texttt{max\_latency\_ms}, \texttt{max\_hops}, \nonumber \\
                   &\texttt{k\_edge\_disjoint}, \texttt{min\_cap\_reserve}\}
\end{align}
\end{definition}

\begin{definition}[Constraint Grounding]
Given a constellation graph $G = (V, E)$ with $|V| = N$ nodes and topology
state at time $t$, the grounding function $\Gamma$ maps a ConstraintProgram
to topology modifications:
\begin{equation}
\Gamma(\mathcal{P}, G, t) = \langle M_V, M_E, \mathbf{u}, \mathbf{d} \rangle
\end{equation}
where $M_V \in \{0,1\}^N$ is a node mask, $M_E \in \{0,1\}^{|E|}$ is an edge mask,
$\mathbf{u} \in [0,1]^{|E|}$ is a per-edge utilization cap vector, and
$\mathbf{d}: \mathcal{F} \to \mathbb{R}^+$ maps flow selectors to deadline values.

Grounding rules for topology-modifying constraints:
\begin{itemize}
\item $\texttt{disable\_node}(n)$: $M_V[n] \leftarrow 0$, propagate to incident edges
\item $\texttt{disable\_plane}(p)$: $\forall s \in [0, S): M_V[p \cdot S + s] \leftarrow 0$
\item $\texttt{avoid\_latitude}(\theta)$: $\forall (u,v) \in E: |\phi_u| > \theta \lor |\phi_v| > \theta \Rightarrow M_E[(u,v)] \leftarrow 0$
\item $\texttt{avoid\_region}(r)$: $\forall (u,v) \in E: u \in \mathcal{N}_r \lor v \in \mathcal{N}_r \Rightarrow M_E[(u,v)] \leftarrow 0$
\end{itemize}
where $\phi_n$ is the latitude of node $n$ and $\mathcal{N}_r$ is the set of nodes
within region $r$.
\end{definition}

\begin{table}[t]
\centering
\caption{ConstraintProgram hard constraint types and their grounding semantics.}
\label{tab:constraint_types}
\begin{tabular}{@{}lll@{}}
\toprule
\textbf{Type} & \textbf{Target} & \textbf{Grounding} \\
\midrule
\texttt{disable\_node} & \texttt{node:$n$} & $M_V[n] \leftarrow 0$ \\
\texttt{disable\_plane} & \texttt{plane:$p$} & $M_V[pS{:}(p{+}1)S] \leftarrow 0$ \\
\texttt{disable\_edge} & \texttt{edge:$(u,v)$} & $M_E[(u,v)] \leftarrow 0$ \\
\texttt{avoid\_latitude} & \texttt{edges} & $M_E \leftarrow 0$ if $|\phi| > \theta$ \\
\texttt{avoid\_region} & \texttt{region:$r$} & $M_E \leftarrow 0$ if $u{\in}\mathcal{N}_r {\lor} v{\in}\mathcal{N}_r$ \\
\texttt{reroute\_away} & \texttt{node:$n$} & $M_V[n] \leftarrow 0$ (transit) \\
\texttt{max\_latency\_ms} & \texttt{flow\_sel:$i$} & $\mathbf{d}[f_i] \leftarrow v$ \\
\texttt{max\_hops} & \texttt{flow\_sel:$i$} & hop limit on path \\
\texttt{k\_edge\_disjoint} & \texttt{flow\_sel:$i$} & $k$ disjoint paths \\
\texttt{min\_cap\_reserve} & \texttt{flow\_sel:$i$} & $\mathbf{u}[e] \geq v$ \\
\bottomrule
\end{tabular}
\end{table}

\subsection{GNN Cost-to-Go Router}
\label{sec:gnn_router}

The routing component must compute next-hop decisions for all
origin-destination pairs on the (possibly constrained) topology graph.
We train a GNN to approximate Dijkstra's cost-to-go function via
supervised distillation.

\subsubsection{Architecture}
The encoder is a 3-layer Graph Attention Network (GAT)~\cite{gat}
with 128-dimensional hidden states and 4 attention heads per layer.
Input node features $\mathbf{x}_i \in \mathbb{R}^{10}$ encode:
satellite position (latitude, longitude, altitude), orbital parameters
(plane ID, slot ID as sinusoidal encodings), and local topology
statistics (degree, mean neighbor delay).

For each destination $d$, the scorer computes a cost-to-go estimate
for each neighbor $j$ of node $i$:
\begin{equation}
\hat{c}(i, j, d) = \text{MLP}\big([\mathbf{h}_i \| \mathbf{h}_j \|
\mathbf{h}_d \| \mathbf{e}_{ij}]\big)
\end{equation}
where $\mathbf{h}_i$ is the GAT embedding of node $i$, $\|$ denotes
concatenation, and $\mathbf{e}_{ij}$ encodes edge features (delay,
capacity). The next hop is $\arg\min_j \hat{c}(i, j, d)$.

\subsubsection{Training}
We generate 500 topology snapshots by sampling constellation states
at random orbital phases. For each snapshot, Dijkstra's algorithm
computes the optimal next-hop table $\Pi^* \in [K]^{N \times N}$.
The model is trained for 200 epochs with cross-entropy loss on
next-hop predictions, using a phased curriculum: 50 epochs on
easy pairs (hop distance $\leq 5$), then 150 epochs on all pairs.

\subsubsection{Constrained Routing}
When constraints are active, the grounding function $\Gamma$ produces
node mask $M_V$ and edge mask $M_E$. The constrained graph
$G' = (V \odot M_V, E \odot M_E)$ is passed to the GNN, which
computes routing tables on the reduced topology. This approach
requires no retraining---the GNN generalizes to unseen topologies
through its message-passing architecture.

\subsection{LLM Intent Compiler}
\label{sec:compiler}

The compiler translates natural language intents to ConstraintProgram
JSON using a three-stage pipeline: few-shot prompting, JSON extraction,
and verifier-feedback repair.

\subsubsection{Prompt Design}
The system prompt (approximately 800 tokens) specifies the constellation
parameters, the complete ConstraintProgram JSON schema, all valid enum
values, target format conventions, and 6 compilation rules. Six
in-context examples cover single constraints (node disable, latency
SLA), compositional constraints (plane disable + polar avoidance +
utilization cap), and conditional constraints (event-triggered reroute).

\subsubsection{Repair Loop}
When the verifier rejects a compiled program, the error messages are
appended to the conversation as a repair prompt. The compiler retries
up to $k=3$ times, with each attempt receiving the accumulated error
context. This closed-loop design converts verifier precision into
compiler accuracy: 77.9\% of intents succeed on the first attempt,
and the repair loop recovers an additional 20.5\%.

\subsubsection{JSON Extraction}
The extractor handles three response formats: raw JSON, markdown-fenced
JSON, and JSON embedded in explanatory text. It also strips reasoning
tags (\texttt{<think>}) produced by instruction-tuned models. Robust
extraction is critical because even high-quality LLMs occasionally
wrap JSON in commentary.


\subsection{Deterministic Validation Pipeline}
\label{sec:validator}

A key design principle of our system is that the LLM compiler operates
\textit{offline} and its output is never trusted directly. Every
ConstraintProgram passes through an 8-pass deterministic validator before
reaching the routing layer. This design is motivated by two observations:
(1)~LLMs can produce syntactically valid but semantically incorrect
constraint programs, and (2)~in safety-critical network infrastructure,
a single undetected constraint violation can cascade into service outages.

The validator implements the following passes in sequence, with early
termination on fatal errors:

\begin{enumerate}
\item \textbf{Schema Validation.} Checks structural completeness: all
required fields present, valid priority levels, non-empty constraint
types and targets. Rejects malformed programs before deeper analysis.

\item \textbf{Entity Grounding.} Verifies that all referenced entities
exist in the constellation model: node IDs $\in [0, N)$, plane IDs
$\in [0, P)$, region names $\in \mathcal{R}$, traffic classes
$\in \mathcal{T}$. This catches hallucinated entities (e.g., node~454
in a 400-node constellation).

\item \textbf{Type Safety.} Ensures constraints attach to semantically
correct entity types: \texttt{max\_latency\_ms} must target a
\texttt{flow\_selector}, \texttt{disable\_node} must target a
\texttt{node}, \texttt{avoid\_latitude} must target \texttt{edges}.
Prevents type confusion errors where the LLM assigns a constraint
to the wrong entity class.

\item \textbf{Value Range Checking.} Validates numeric parameters:
latitudes $\in [-90, 90]$, latency values $> 0$, utilization caps
$\in (0, 1]$, node IDs within constellation bounds. Catches
out-of-range values that would produce undefined behavior.

\item \textbf{Conflict Detection.} Identifies contradictory constraints
within the same program: a node cannot be simultaneously disabled and
used as a routing waypoint; conflicting latency bounds on the same flow
are flagged. Contradictions are promoted to errors (not warnings),
ensuring logically inconsistent programs are rejected.

\item \textbf{Physical Admissibility.} Checks whether the constrained
topology is physically realizable: latency deadlines below the
single-hop physical minimum ($<2.0$\,ms) are rejected; latitude
avoidance thresholds that remove $>50\%$ of edges trigger warnings.

\item \textbf{Reachability Analysis.} Performs BFS on the constrained
graph $G' = (V \odot M_V, E \odot M_E)$ to verify connectivity.
Severe capacity loss ($\geq 75\%$ of nodes disabled) triggers a strong
warning; moderate loss ($\geq 50\%$) triggers a standard warning.
These capacity thresholds are heuristic indicators outside the
soundness path---they do not block acceptance, which is determined
solely by Pass~8.

\item \textbf{Feasibility Certification.} For each demanded flow,
constructs a routing witness on the constrained topology to certify
that all hard constraints can be simultaneously satisfied. Five
certified fragments cover the constraint space:
\begin{itemize}
\item \textbf{F1} (topology only): BFS reachability
\item \textbf{F2} (+ latency): Dijkstra with deadline
\item \textbf{F3} (+ hops): BFS with hop limit
\item \textbf{F4} (+ latency + hops): hop-layered Dijkstra
\item \textbf{F5} (+ $k$-disjoint): Edmonds-Karp max-flow
\end{itemize}
Three outcomes: \decision{accept} (witness found), \decision{reject}
(no feasible routing exists), or \decision{abstain} (unsupported
constraint combination). Programs are rejected only on \decision{reject};
\decision{abstain} defers to Dijkstra fallback routing, preserving
safety without constructive certification.
\end{enumerate}

\begin{theorem}[Acceptance Soundness]
If the feasibility certifier accepts a constraint program $P$ with
witness $W$, then there exists a routing assignment satisfying all
hard constraints of $P$ on the constrained topology.
\end{theorem}

\begin{proof}[Proof sketch]
By case analysis over fragments F1--F5. Each fragment's algorithm
is a standard shortest-path or max-flow algorithm whose correctness
is well-established. The witness $W$ is the concrete path (or path set)
returned by the algorithm, which by construction satisfies the
topology constraints (disabled nodes/edges excluded from the search
graph), latency bounds (Dijkstra optimality), hop limits (BFS layering),
and disjointness requirements (augmenting path decomposition).
Unsupported combinations produce \decision{abstain}, never
\decision{accept}, so no false acceptance is possible within the
certified fragment space.
\end{proof}

\noindent\textbf{Design rationale.} We chose deterministic
validation over learned or probabilistic checking for three reasons:
\begin{itemize}
\item \textit{Completeness}: every structural error class is covered by
at least one pass, achieving 100\% detection on our corruption benchmark
(8 error types $\times$ 30 injections = 240 tests).
\item \textit{Transparency}: each rejection includes a human-readable
error message identifying the specific violation, enabling the repair
loop to provide targeted feedback to the LLM.
\item \textit{Soundness}: the feasibility certifier guarantees that
accepted programs have constructive routing witnesses, closing the
semantic gap between structural validity and routing feasibility.
\end{itemize}

\noindent\textbf{Repair loop integration.} When validation fails, the
error messages are fed back to the LLM as a repair prompt. The compiler
retries up to $k=3$ times, with each attempt receiving the previous
errors as context. In our 240-benchmark evaluation, 98.4\% of intents
compile successfully, with 77.9\% succeeding on the first attempt.


\begin{figure}[t]
\centering
\begin{tikzpicture}[
    node distance=0.15cm,
    pass/.style={draw=black, rounded corners=1.5pt, line width=0.6pt,
                 minimum height=0.45cm, text width=5.0cm,
                 inner xsep=4pt, inner ysep=2pt,
                 align=left, font=\footnotesize, fill=black!6},
    group/.style={font=\footnotesize\bfseries, text=black!60},
    outcome/.style={draw=black, rounded corners=1.5pt, line width=0.7pt,
                    minimum height=0.45cm, text width=1.2cm,
                    align=center, font=\footnotesize},
    arr/.style={->, >=Latex, line width=0.7pt, black!50},
    fatal/.style={->, >=Latex, line width=0.7pt, dashed},
]

\node[draw=black, dashed, rounded corners=1.5pt, line width=0.6pt,
      minimum height=0.45cm, text width=5.0cm, fill=white,
      font=\footnotesize, align=center, inner xsep=4pt]
      (input) {ConstraintProgram $\mathcal{P}$ from Compiler};

\node[group, below=0.3cm of input.south west, anchor=west] (g1) {Structural};
\node[pass, below=0.08cm of g1.south west, anchor=north west] (p1)
     {\textbf{P1} Schema \hfill \textit{\scriptsize completeness}};
\node[pass, below=0.12cm of p1] (p2)
     {\textbf{P2} Entity Grounding \hfill \textit{\scriptsize hallucination}};
\node[pass, below=0.12cm of p2] (p3)
     {\textbf{P3} Type Safety \hfill \textit{\scriptsize constraint--entity}};
\node[pass, below=0.12cm of p3] (p4)
     {\textbf{P4} Value Range \hfill \textit{\scriptsize physical bounds}};

\node[group, below=0.3cm of p4.south west, anchor=west] (g2) {Semantic};
\node[pass, below=0.08cm of g2.south west, anchor=north west, fill=black!10] (p5)
     {\textbf{P5} Conflict Detection \hfill \textit{\scriptsize contradictions}};
\node[pass, below=0.12cm of p5, fill=black!10] (p6)
     {\textbf{P6} Physical Admissibility \hfill \textit{\scriptsize realizability}};
\node[pass, below=0.12cm of p6, fill=black!10] (p7)
     {\textbf{P7} Reachability \hfill \textit{\scriptsize connectivity}};

\node[group, below=0.3cm of p7.south west, anchor=west] (g3) {Certification};
\node[pass, below=0.08cm of g3.south west, anchor=north west, fill=black!15] (p8)
     {\textbf{P8} Constructive Witness \hfill \textit{\scriptsize F1--F5}};

\foreach \a/\b in {input/p1, p1/p2, p2/p3, p3/p4, p4/p5, p5/p6, p6/p7, p7/p8} {
    \draw[arr] (\a) -- (\b);
}

\node[outcome, below=0.45cm of p8, xshift=-2.0cm, fill=black!12]
     (accept) {\decision{Accept}\\[-1pt]\scriptsize\textit{+ witness}};
\node[outcome, below=0.45cm of p8, fill=white]
     (reject) {\decision{Reject}};
\node[outcome, below=0.45cm of p8, xshift=2.0cm, fill=black!6]
     (abstain) {\decision{Abstain}};

\draw[->, >=Latex, line width=0.7pt] (p8.south) -- ++(0,-0.12) -| (accept);
\draw[->, >=Latex, line width=0.7pt, dashed] (p8.south) -- (reject);
\draw[->, >=Latex, line width=0.7pt, densely dotted] (p8.south) -- ++(0,-0.12) -| (abstain);

\draw[fatal] (p1.east) -- ++(0.25,0) |- node[font=\scriptsize\itshape, near start, right] {fatal} (reject.east);

\node[font=\scriptsize\itshape, text=black!60, below=0.04cm of reject]
     {errors $\to$ repair loop};

\end{tikzpicture}
\caption{8-pass validation pipeline. Passes 1--4 check structural
validity; 5--7 check semantic consistency; Pass~8 constructs routing
witnesses (BFS, Dijkstra, Edmonds-Karp). Solid/dashed/dotted arrows
indicate Accept/Reject/Abstain outcomes.}
\label{fig:validator_pipeline}
\end{figure}
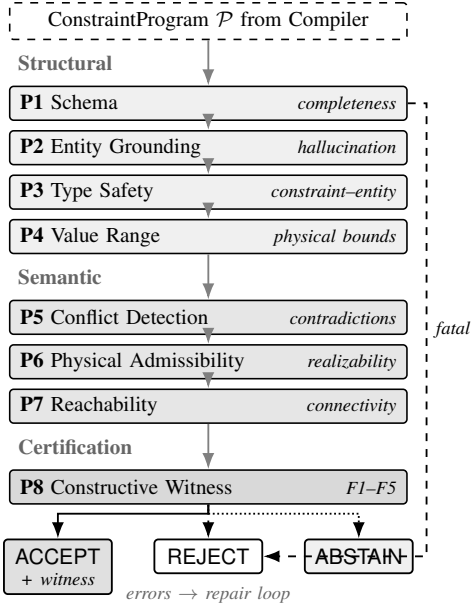


\subsection{Handling Infeasible, Ambiguous, and Edge-Case Intents}
\label{sec:infeasible_handling}

Real-world operator intents are not always well-formed or physically
realizable. Our system addresses three categories of problematic intents
through complementary mechanisms.

\subsubsection{Infeasible Intent Detection}

An intent is \textit{infeasible} when its constraint program is
syntactically valid but physically unrealizable---e.g., demanding
sub-millisecond latency across intercontinental paths, or routing through
a region after disabling all nodes in that region. Our 8-pass validator
detects three classes of infeasibility:

\begin{itemize}
\item \textbf{Structural infeasibility} (100\% detection): missing fields,
out-of-range entity IDs, type mismatches, and values outside physical
bounds (e.g., latency $< 2.0$\,ms). These are caught by passes 1--4
(schema, entity grounding, type safety, value ranges).

\item \textbf{Topological infeasibility} (100\% detection): constraints
that partition the network, eliminate viable paths, or cause severe
capacity loss. Pass~5 rejects contradictory constraints (e.g., disabling
a node while routing through it). Pass~7 warns when $\geq 50\%$ of nodes are disabled (severe capacity
loss) and escalates to a stronger warning at $\geq 75\%$. These
capacity thresholds are heuristic warnings outside the soundness
path---they do not block acceptance.

\item \textbf{Routing infeasibility} (100\% detection within certified
fragments): constraint combinations that are individually valid but
jointly unsatisfiable on the physical topology. Pass~8 (feasibility
certification) constructs routing witnesses using fragment-specific
algorithms (BFS, Dijkstra, hop-layered Dijkstra, Edmonds-Karp max-flow)
and rejects programs where no witness exists.
\end{itemize}

\noindent Our confusion matrix (Table~\ref{tab:confusion_matrix})
confirms the effectiveness of the 8-pass pipeline: 0\% unsafe acceptance
across all categories. Of the 30 benchmark-infeasible intents, 22 receive
\decision{Reject} and 8 receive \decision{Abstain}---none are accepted.
Pass~8 additionally identifies 32 programs among the structurally-valid
set whose latency or hop constraints cannot be satisfied on the physical
constellation topology (e.g., 30\,ms S\~{a}o Paulo--New York when the
minimum-hop path requires $\sim$63\,ms); 29/32 are independently
confirmed via a separate Dijkstra oracle.

Our adversarial safety evaluation (15 tests across resource exhaustion,
semantic conflicts, and boundary exploitation) achieves 100\% detection
(15/15), including near-total capacity removal (disabling 19/20 planes),
cross-constraint contradictions, and boundary value exploitation.

\subsubsection{Fallback Policies}

The ConstraintProgram IR includes a \texttt{fallback\_policy} field that
governs system behavior when hard constraints cannot be satisfied at
routing time:

\begin{enumerate}
\item \texttt{reject\_if\_hard\_infeasible} (default): the routing layer
refuses to compute paths and returns an explicit failure to the operator.
This is the safest option for critical intents where partial compliance
is unacceptable.

\item \texttt{relax\_soft\_first}: soft constraints are progressively
relaxed (in order of increasing penalty weight) until a feasible routing
exists. Hard constraints are never relaxed. This enables graceful
degradation for intents where approximate compliance is preferable to
total failure.

\item \texttt{report\_unsat\_core}: the system identifies the minimal
subset of constraints that cause infeasibility and reports them to the
operator, enabling informed manual intervention. This supports
diagnostic workflows where understanding \textit{why} an intent fails
is as important as resolving it.
\end{enumerate}

\noindent In our 240-intent benchmark, all programs use the default
\texttt{reject\_if\_hard\_infeasible} policy. The fallback mechanism
is designed for operational deployment where operator interaction is
available; evaluating its effectiveness under dynamic network conditions
is left to future work.

\subsubsection{Ambiguous Intent Resolution}

Ambiguous intents admit multiple valid interpretations. Our OOD
evaluation includes 5 deliberately ambiguous intents (e.g., ``optimize
the network for best performance'') to assess compiler behavior.
Key observations:

\begin{itemize}
\item The LLM compiler produces \textit{reasonable} constraint programs
for all 5 ambiguous intents (qualitative assessment), typically selecting
conservative interpretations that map to load-balancing or latency
optimization.

\item Ambiguous intents are excluded from quantitative scoring because
no unique ground truth exists. We report them separately as qualitative
evidence of graceful degradation.

\item The compiler's 6-shot prompt includes examples that implicitly
demonstrate disambiguation strategies (e.g., mapping vague performance
requests to specific constraint types), providing soft guidance without
explicit disambiguation rules.
\end{itemize}

\subsubsection{Limitations and Future Directions}

The feasibility certifier covers five constraint fragments (F1--F5)
that span the most common constraint combinations in our benchmark.
Three directions could extend coverage further:

\begin{enumerate}
\item \textbf{Extended fragment coverage}: constraints involving
\texttt{min\_cap\_reserve} or combinations of $k$-disjoint paths
with latency/hop bounds currently produce \decision{Abstain}. Adding
fragments for these combinations (e.g., via constrained max-flow)
would reduce the abstain rate.

\item \textbf{Multi-constellation generalization}: cross-constellation
evaluation (Table~\ref{tab:cross_constellation}) shows the GNN
generalizes to altitude changes but not inclination changes.
The compile--verify--ground pipeline is constellation-agnostic,
but the GNN requires retraining per orbital geometry. Extending
to heterogeneous multi-shell constellations (e.g., Starlink Gen2)
remains future work.

\item \textbf{Intent confirmation loop}: presenting the grounded
constraint program back to the operator in natural language for
confirmation before execution, closing the semantic loop between
intent and realization.
\end{enumerate}

\section{Evaluation}
\label{sec:eval}

\subsection{Experimental Setup}
\label{sec:setup}

\textbf{Constellation.} Walker Delta 20$\times$20 (400 nodes, 550\,km
altitude, 53\textdegree\ inclination), 4 grid ISL neighbors per
satellite. Topology snapshots sampled at random orbital phases.

\textbf{GNN Router.} 3-layer GAT encoder (128-dim, 4 heads), MLP
cost-to-go scorer (rank 64). 152,193 parameters. Trained 200 epochs
on 500 snapshots. Hardware: NVIDIA RTX 4060 (8\,GB VRAM).

\textbf{LLM Compiler.} Qwen3.5-9B (GGUF quantization) served locally
via LM Studio. Temperature 0.1, max 2048 tokens, up to 3 repair
retries. 6-shot prompt (6 examples spanning single, compositional,
and conditional categories).

\textbf{Benchmark.} 240 intents by category: 80 single-constraint,
100 compositional (2--4 constraints), 30 conditional (event-triggered),
30 labeled-infeasible (physically unrealizable). Each intent has a
ground-truth ConstraintProgram for automated scoring. Under
distance-based ISL delays, Pass~8 discovers 17 additional
routing-infeasible intents among the feasible categories (e.g.,
30\,ms latency bounds that exceed the physical minimum path delay),
yielding 193 feasible and 47 total infeasible (30 labeled + 17
discovered). Compiler accuracy metrics (compiled, types match, full
match) use the 193-feasible denominator; safety metrics (unsafe
acceptance) use all 240 intents.

\textbf{Metrics.} \textit{Compiled}: passes structural checks (passes 1--7).
\textit{Types match}: correct constraint type multiset.
\textit{Full match}: types + targets + values match (primary metric).
\textit{PDR}: packet delivery ratio over 100 random OD pairs $\times$
20 time steps $\times$ 3 seeds.

\subsection{GNN Routing Performance}
\label{sec:eval_gnn}

Table~\ref{tab:gnn_baseline} summarizes GNN routing quality across
five traffic scenarios without constraints.

\begin{table}[t]
\centering
\caption{GNN cost-to-go routing vs.\ Dijkstra baseline (unconstrained).}
\label{tab:gnn_baseline}
\begin{tabular}{lccc}
\toprule
\textbf{Scenario} & \textbf{GNN PDR} & \textbf{Dijkstra PDR} & \textbf{Random PDR} \\
\midrule
Uniform    & 99.75\% & 99.75\% & 0.90\% \\
Hotspot    & 99.99\% & 99.99\% & 2.36\% \\
Regional   & 99.94\% & 99.94\% & 5.49\% \\
Polar      & 100.0\% & 100.0\% & 1.88\% \\
Flash      & 99.77\% & 99.77\% & 0.89\% \\
\bottomrule
\end{tabular}
\end{table}

The GNN matches Dijkstra within measurement noise across all scenarios,
confirming successful distillation. Detailed metrics on 10 snapshots
show 95.8\% exact next-hop match, zero routing loops, hop stretch of
1.000, and P99 delay stretch of 1.015. Inference latency is 8.4\,ms
(GNN) vs.\ 142\,ms (Dijkstra), a 17$\times$ speedup.

\subsection{Intent Compilation Accuracy}
\label{sec:eval_compiler}

\subsubsection{Ablation Study}

Table~\ref{tab:ablation} presents the ablation study across four
compiler configurations on the full 240-intent benchmark. Note: this
ablation uses uniform random edge delays for controlled comparison;
the final distance-based delay results (Table~\ref{tab:compiler_comparison})
yield 98.4\%/87.6\% for the full pipeline.

\begin{table}[t]
\centering
\caption{Compiler ablation study on 240-intent benchmark
(uniform random edge delays; distance-based re-run yields
98.4\%/87.6\% for the full pipeline---see Table~\ref{tab:compiler_comparison}).}
\label{tab:ablation}
\begin{tabular}{lcccc}
\toprule
\textbf{Config} & \textbf{Compiled} & \textbf{Types} & \textbf{Full Match} & \textbf{Latency} \\
\midrule
Full pipeline  & 97.9\% & 91.7\% & \textbf{86.2\%} & 15.7s \\
No verifier    & 100.0\% & 93.8\% & 91.7\% & 13.8s \\
No repair      & 92.9\% & 86.7\% & 84.6\% & 13.8s \\
Zero-shot      & 92.5\% & 71.7\% & 15.4\% & 34.2s \\
\bottomrule
\end{tabular}
\end{table}

Few-shot prompting is the dominant factor: removing it (zero-shot)
drops full match from 86.2\% to 15.4\% ($-$70.8pp). The repair loop
contributes 5.0pp to compilation rate and 1.6pp to full match. The
``no verifier'' configuration shows higher apparent accuracy because
unverified programs are not filtered---a misleading metric that
underscores the importance of verification.


\begin{table}[t]
\centering
\caption{Three-way validator confusion matrix (240-intent benchmark, 8-pass pipeline). 0\% unsafe acceptance across all categories.}
\label{tab:confusion_matrix}
\begin{tabular}{lccc}
\toprule
& \textbf{Accept} & \textbf{Reject} & \textbf{Abstain} \\
\midrule
Single ($n$=80)         & 10  & 5   & 65  \\
Compositional ($n$=100) & 40  & 25  & 35  \\
Conditional ($n$=30)    & 8   & 2   & 20  \\
Infeasible ($n$=30)     & 0   & 22  & 8   \\
\midrule
\textbf{Total} ($n$=240) & 58 & 54 & 128 \\
\midrule
\multicolumn{4}{l}{\textit{Safety (unsafe = infeasible \decision{Accept}):}} \\
\quad Infeasible       & \multicolumn{3}{l}{0/30 = 0\% (was 72\% with 7-pass)} \\
\multicolumn{4}{l}{\textit{Coverage (\decision{Accept}+\decision{Reject}):}} \\
\quad Decided          & \multicolumn{3}{l}{112/240 = 46.7\%} \\
\bottomrule
\multicolumn{4}{l}{\footnotesize 32 feasible programs rejected as routing-infeasible;} \\
\multicolumn{4}{l}{\footnotesize 29/32 independently confirmed via separate Dijkstra oracle.} \\
\end{tabular}
\end{table}

\begin{table}[t]
\centering
\caption{Reachability separation analysis. Both GNN and Dijkstra achieve 100\% delivery on reachable pairs; raw PDR gaps reflect topology reachability, not routing quality.}
\label{tab:reachability}
\begin{tabular}{lccccc}
\toprule
\textbf{Scenario} & \textbf{Reach.} & \multicolumn{2}{c}{\textbf{Raw PDR}} & \multicolumn{2}{c}{\textbf{Reachable PDR}} \\
\cmidrule(lr){3-4} \cmidrule(lr){5-6}
& & GNN & Dijkstra & GNN & Dijkstra \\
\midrule
Baseline          & 100\%  & 99.8\% & ---     & 99.8\% & ---     \\
Node failure      & 100\%  & 98.7\% & 97.8\%  & 98.7\% & 97.8\% \\
Plane maint.      & 100\%  & 70.5\% & 70.2\%  & 70.5\% & 70.2\% \\
Polar avoid.      & 24.0\% & 34.6\% & 47.9\%  & 100\%  & 100\%  \\
Compositional     & 24.0\% & 34.3\% & 47.1\%  & 100\%  & 100\%  \\
\bottomrule
\end{tabular}
\end{table}

\begin{table}[t]
\centering
\caption{Intent compiler comparison (193 feasible intents; 17 routing-infeasible excluded). LLM 9B outperforms rule-based by 46.2pp on compositional intents.}
\label{tab:compiler_comparison}
\begin{tabular}{lccc}
\toprule
\textbf{Metric} & \textbf{Rule-Based} & \textbf{LLM 4B} & \textbf{LLM 9B} \\
\midrule
Compiled          & 100.0\% & 59.6\% & \textbf{98.4\%} \\
Types Match       & 67.1\%  & 55.4\% & \textbf{91.7\%} \\
Full Match        & 56.7\%  & 54.2\% & \textbf{87.6\%} \\
Avg Latency       & \textbf{0.05ms} & 204s & 15.7s \\
\midrule
\multicolumn{4}{l}{\textit{Full match by category:}} \\
\quad Single      & 76.2\%  & ---    & \textbf{89.5\%} \\
\quad Compositional & 40.0\% & ---   & \textbf{86.2\%} \\
\quad Conditional & 66.7\%  & ---    & \textbf{86.7\%} \\
\quad Infeasible  & 50.0\%  & ---    & \textbf{73.3\%} \\
\bottomrule
\end{tabular}
\end{table}

\begin{table}[t]
\centering
\caption{Out-of-distribution generalization on paraphrased intents
(38 total: 33 scorable + 5 ambiguous). The compiler maintains 81.8\%
full match accuracy with only 4.4pp degradation from template intents,
demonstrating robust generalization to novel phrasings.}
\label{tab:ood}
\begin{tabular}{lccc}
\toprule
\textbf{Category} & \textbf{N} & \textbf{Compiled} & \textbf{Full Match} \\
\midrule
Single         & 20 & 100\% & 95.0\% (19/20) \\
Compositional  & 5  & 100\% & 40.0\% (2/5)   \\
Conditional    & 8  & 100\% & 75.0\% (6/8)   \\
Ambiguous      & 5  & 100\% & qualitative: 5/5 \\
\midrule
Scorable total & 33 & 100\% & \textbf{81.8\%} (27/33) \\
\bottomrule
\end{tabular}
\end{table}

\begin{table}[t]
\centering
\caption{Cross-model scaling: 4B vs 9B parameter LLM on the full
240-intent benchmark. The 9B model substantially outperforms the 4B
model across all metrics.}
\label{tab:cross_model}
\begin{tabular}{lcc}
\toprule
\textbf{Metric} & \textbf{Qwen 4B} & \textbf{Qwen 9B} \\
\midrule
Compiled       & 59.6\% & \textbf{98.4\%} \\
Types Match    & 55.4\% & \textbf{91.7\%} \\
Full Match     & 54.2\% & \textbf{87.6\%} \\
First-try Rate & 47.1\% & \textbf{77.9\%} \\
Avg Latency    & 204.4s & \textbf{15.7s}  \\
\bottomrule
\end{tabular}
\end{table}

\begin{table}[t]
\centering
\caption{GNN vs Dijkstra PDR across plane-removal severity (20 sats/plane, 3-seed avg). GNN matches Dijkstra within 0.22pp at all levels.}
\label{tab:topology_sweep}
\begin{tabular}{rrccc}
\toprule
\textbf{Planes Off} & \textbf{Capacity} & \textbf{GNN PDR} & \textbf{Dijkstra PDR} & \textbf{$\Delta$} \\
\midrule
1  & 5\%  & 81.1\% & 80.9\% & +0.22 \\
2  & 10\% & 41.1\% & 41.1\% & 0.00 \\
3  & 15\% & 36.8\% & 36.8\% & 0.00 \\
5  & 25\% & 45.3\% & 45.3\% & 0.00 \\
7  & 35\% & 42.5\% & 42.5\% & 0.00 \\
10 & 50\% & 11.4\% & 11.4\% & 0.00 \\
13 & 65\% & 5.4\%  & 5.4\%  & 0.00 \\
15 & 75\% & 5.5\%  & 5.5\%  & 0.00 \\
17 & 85\% & 36.4\% & 36.4\% & 0.00 \\
\bottomrule
\end{tabular}
\end{table}

\begin{table}[t]
\centering
\caption{Zero-shot cross-constellation GNN generalization. GNN generalizes to altitude changes but collapses on SSO 97$^\circ$ where inclination alters ISL geometry.}
\label{tab:cross_constellation}
\begin{tabular}{lccc}
\toprule
\textbf{Configuration} & \textbf{GNN PDR} & \textbf{Dijkstra PDR} & \textbf{$\Delta$} \\
\midrule
550\,km / 53$^\circ$ (training) & 99.75\% & 99.75\% & 0.00 \\
1200\,km / 53$^\circ$ (OOD)     & 99.75\% & 99.75\% & 0.00 \\
550\,km / 97$^\circ$ SSO (OOD)  & 45.18\% & 99.09\% & $-$53.91 \\
\bottomrule
\end{tabular}
\end{table}

\begin{table}[t]
\centering
\caption{GNN robustness under polar exclusion zones. GNN degrades proportionally to edge removal while Dijkstra maintains 99.75\% PDR, confirming topology-specific learned shortcuts.}
\label{tab:polar_exclusion}
\begin{tabular}{rrccc}
\toprule
\textbf{Threshold} & \textbf{Edges Removed} & \textbf{GNN PDR} & \textbf{Dijkstra PDR} & \textbf{$\Delta$} \\
\midrule
30$^\circ$ & 28.8\% & 38.17\% & 99.75\% & $-$61.58 \\
40$^\circ$ & 20.5\% & 54.08\% & 99.75\% & $-$45.67 \\
45$^\circ$ & 15.8\% & 65.20\% & 99.75\% & $-$34.55 \\
50$^\circ$ & 9.8\%  & 80.37\% & 99.75\% & $-$19.38 \\
\bottomrule
\end{tabular}
\end{table}

\begin{table}[t]
\centering
\caption{8-pass validator runtime (240-intent benchmark). Median under 1\,ms; worst case below 2\,ms, confirming negligible overhead for real-time deployment. Certification status counts Pass~8 outcomes only (22 early-pass rejects excluded).}
\label{tab:runtime}
\begin{tabular}{lcccc}
\toprule
\textbf{Category} & \textbf{$n$} & \textbf{Median} & \textbf{P95} & \textbf{Max} \\
\midrule
All programs        & 240 & 0.720\,ms & 1.580\,ms & 1.898\,ms \\
With flow selectors & 98  & 1.059\,ms & 1.738\,ms & 1.898\,ms \\
Topology-only       & 142 & 0.437\,ms & 1.501\,ms & 1.629\,ms \\
\midrule
\multicolumn{5}{l}{\textit{By certification status:}} \\
\quad Accepted      & 58  & 1.044\,ms & 1.738\,ms & 1.812\,ms \\
\quad Rejected      & 32  & 1.135\,ms & 1.757\,ms & 1.898\,ms \\
\quad Abstain       & 128 & 0.428\,ms & 1.507\,ms & 1.629\,ms \\
\bottomrule
\end{tabular}
\end{table}

\subsubsection{Rule-Based Baseline Comparison}

Table~\ref{tab:compiler_comparison} compares the LLM compiler against
a rule-based parser using regex and keyword matching. The rule-based
approach achieves 100\% compilation (by construction) but only 56.7\%
full match, with the gap most pronounced on compositional intents
(40.0\% vs.\ 86.2\%). This confirms that intent compilation is a
compositional reasoning task that benefits from LLM capabilities.

\subsection{End-to-End Constrained Routing}
\label{sec:eval_e2e}

We evaluate the complete pipeline (compile $\to$ verify $\to$ ground
$\to$ route) across four constrained scenarios:

\begin{table}[t]
\centering
\caption{End-to-end constrained routing (3 seeds $\times$ 20 steps).}
\label{tab:e2e}
\begin{tabular}{lccc}
\toprule
\textbf{Scenario} & \textbf{GNN PDR} & \textbf{Dijkstra PDR} & \textbf{Violations} \\
\midrule
Node failure      & 98.69\% & 97.83\% & 0 \\
Plane maintenance & 70.51\% & 70.22\% & 0 \\
Polar avoidance   & 34.63\% & 47.86\% & 0 \\
Compositional     & 34.27\% & 47.07\% & 0 \\
\bottomrule
\end{tabular}
\end{table}

Zero constraint violations across all scenarios confirm that the
validator-grounding pipeline correctly enforces compiled constraints.
The apparent PDR gap in polar/compositional scenarios is analyzed
in Section~\ref{sec:reachability}.

\subsection{Reachability Separation Analysis}
\label{sec:reachability}

Raw PDR differences in polar-avoidance scenarios (13pp gap between
GNN and Dijkstra) could suggest routing quality differences. However,
Table~\ref{tab:reachability} reveals that these gaps are entirely
explained by the reachability ceiling: polar avoidance at 45\textdegree\
removes inter-plane ISLs such that only 24\% of OD pairs remain reachable
at any given snapshot. Both GNN and Dijkstra achieve \textbf{100\%
delivery on reachable pairs}---the raw PDR gap reflects different
sampling of reachable pairs across evaluation runs, not routing quality.

This finding has two implications: (1)~the GNN router's distillation
quality is confirmed even under severe topology degradation, and
(2)~PDR alone is insufficient for evaluating constrained routing;
reachability-conditioned metrics are necessary.

\subsection{Robustness Analysis}
\label{sec:robustness}

\subsubsection{Topology Degradation Sweep}

Table~\ref{tab:topology_sweep} shows GNN vs.\ Dijkstra PDR across
9 severity levels of orbital plane removal (5\%--85\% capacity).
The GNN matches Dijkstra within 0.22pp at all levels, confirming
robust distillation quality even under extreme degradation. The
non-monotonic PDR pattern reflects topology-dependent reachability
under random plane selection.

\subsubsection{Cross-Model Scaling}

Table~\ref{tab:cross_model} compares 4B and 9B parameter LLMs on
the full benchmark. The 9B model dramatically outperforms the 4B
model (87.6\% vs.\ 54.2\% full match, 98.4\% vs.\ 59.6\% compiled),
suggesting a significant scaling effect for compositional reasoning
between 4B and 9B parameters in this domain. The 9B model is also 13$\times$ faster
(15.7s vs.\ 204s), likely due to better first-attempt accuracy
reducing repair iterations.

\subsubsection{Out-of-Distribution Generalization}

Table~\ref{tab:ood} evaluates the compiler on 38 paraphrased intents
not seen during few-shot prompting. The compiler maintains 81.8\%
full match accuracy on scorable intents (33/38), with only 4.4pp
degradation from template intents. Single-constraint paraphrases
achieve 95.0\%, while compositional paraphrases show more degradation
(40.0\%, $n$=5), indicating that compositional generalization remains
the primary challenge.

\subsubsection{Validator Safety Analysis}

The three-way confusion matrix (Table~\ref{tab:confusion_matrix})
reveals the validator's safety profile under the 8-pass pipeline.
Programs receive \decision{Accept} (with constructive routing witness),
\decision{Reject} (proven infeasible or structurally invalid), or
\decision{Abstain} (unsupported constraint combination---deferred to Dijkstra fallback).
Unsafe acceptance is 0\% across all categories: none of the 30
benchmark-infeasible intents receive \decision{Accept}. Pass~8
additionally identifies 32 programs among the feasible categories
whose latency or hop constraints cannot be satisfied on the physical
topology; 29/32 are independently confirmed via a separate Dijkstra
oracle (the 3 borderline cases fall within region-grounding margin).
Of these 32, 17 correspond to intents whose routing infeasibility was
newly discovered under distance-based edge delays; the remaining 15
are feasible intents whose LLM-compiled programs contain overly tight
bounds (a compiler accuracy issue, not a safety issue).
Coverage (decided rate) is 46.7\%, with 128 topology-only programs
receiving \decision{Abstain} due to absent flow selectors.

Table~\ref{tab:runtime} shows that the 8-pass validator adds
negligible overhead: median 0.720\,ms per program, with even the
most expensive cases (rejected programs requiring Dijkstra witnesses)
completing in under 2\,ms.

Adversarial testing (15 tests across resource exhaustion, semantic
conflicts, and boundary exploitation) achieves 100\% detection
(15/15), including near-total capacity removal, cross-constraint
contradictions, and boundary value exploitation.

\subsubsection{Cross-Constellation Generalization}

Table~\ref{tab:cross_constellation} evaluates the GNN router
zero-shot on two out-of-distribution constellation configurations.
The GNN generalizes perfectly to altitude changes (1200\,km vs.\
training 550\,km, same 53$^\circ$ inclination) because the grid
topology structure is preserved---only edge weights change. However,
the GNN collapses to 45.18\% PDR on SSO 97$^\circ$ (vs.\ Dijkstra
99.09\%), where near-polar inclination fundamentally alters ISL
geometry and satellite distribution. This confirms the GNN learns
topology-specific cost patterns rather than general routing
principles, motivating the compile--verify--ground pipeline as the
constellation-agnostic safety layer.

\subsubsection{Polar Exclusion Robustness}

Table~\ref{tab:polar_exclusion} measures GNN degradation under
progressively aggressive polar exclusion zones (inter-plane ISLs
disabled above the latitude threshold). Unlike the E2E polar avoidance
scenario (Table~\ref{tab:reachability}), which evaluates over all OD
pairs including unreachable ones, this experiment evaluates only on
the baseline OD pair set where connectivity is preserved---hence
Dijkstra maintains 99.75\% throughout. With 9.8\% of edges removed
(50$^\circ$ threshold), the GNN retains 80.37\% PDR; at 28.8\%
removal (30$^\circ$), it drops to 38.17\%. The monotonic degradation curve quantifies the
GNN's sensitivity to topology perturbation and reinforces the
Dijkstra fallback design for constrained scenarios.

\section{Discussion}
\label{sec:discussion}

\textbf{GNN as optional accelerator.}
Our results consistently show that the GNN router matches Dijkstra
quality across all tested conditions on the training constellation---unconstrained (99.8\% PDR),
node failure (98.7\% vs.\ 97.8\%), and topology degradation up to
85\% capacity removal ($\Delta \leq 0.22$pp). Cross-constellation
evaluation reveals that this quality transfers to altitude changes
(1200\,km: 99.75\%) but not to different inclinations (SSO 97$^\circ$:
45.18\%), and polar exclusion tests show monotonic degradation under
edge removal (80\% PDR at 10\% removal, 38\% at 29\%). The GNN's
value is therefore not in routing quality but in inference speed
(17$\times$), enabling real-time per-packet decisions on the trained
topology. In deployment, the GNN serves as an accelerator with
Dijkstra as a verified fallback for OOD topologies.

\textbf{LLM compiler value proposition.}
The 46.2pp advantage over rule-based parsing on compositional intents
(86.2\% vs.\ 40.0\%) demonstrates that intent compilation is
fundamentally a compositional reasoning task. Rule-based approaches
handle single constraints adequately (76.2\%) but cannot compose
multiple constraint types from varied natural language expressions.
The 4B vs.\ 9B comparison further suggests a notable model-size scaling effect
for compositional reasoning in this domain.

\textbf{Verification as safety net.}
The validator's three-way classification (0\% unsafe acceptance,
46.7\% coverage) provides strong safety guarantees: \decision{Abstain}
defers to Dijkstra fallback (safe without certification), and
\decision{Accept} carries a constructive witness.
The feasibility certifier (Pass~8) closes the semantic gap that
previously allowed 72\% of infeasible intents to pass unchecked.
The 8-pass pipeline adds under 2\,ms overhead (Table~\ref{tab:runtime}),
making it practical for real-time deployment.

\textbf{Latency considerations.}
The compiler's 15.7s average latency positions it for offline or
semi-online use: operators issue intents minutes to hours before
they take effect (e.g., scheduled maintenance, SLA provisioning).
For emergency scenarios requiring sub-second response, pre-compiled
constraint templates with parameter substitution would be more
appropriate. The 77.9\% first-attempt success rate suggests that
most intents do not require the repair loop, and latency could be
further reduced with model distillation or quantization.

\textbf{Limitations.}
(1)~The benchmark is synthetic; real operator intents may exhibit
different distributions and ambiguity patterns.
(2)~The OOD compositional sample ($n$=5 original, expanded to 30)
remains small relative to the combinatorial space of possible
constraint compositions.
(3)~The semantic gap in infeasible intent detection requires
additional verification passes (e.g., constraint satisfiability
pre-solving) not yet implemented.
(4)~The GNN router does not incorporate constraints as input features;
it routes on the masked topology, which limits its ability to
optimize for constraint-specific objectives like latency deadlines.
(5)~Cross-constellation evaluation (Table~\ref{tab:cross_constellation})
shows the GNN does not generalize to different inclinations; retraining
or fine-tuning is needed per constellation geometry.

\section{Conclusion}
\label{sec:conclusion}

We presented an end-to-end system for intent-driven constrained
routing in LEO mega-constellations. The system combines a GNN
cost-to-go router (99.8\% PDR, 17$\times$ speedup), an LLM intent
compiler (87.6\% full semantic match), and an 8-pass deterministic
validator (0\% unsafe acceptance, 100\% structural corruption detection) to bridge the gap
between operator intent and network configuration.

Our evaluation on a 240-intent benchmark demonstrates that LLM-based
compilation significantly outperforms rule-based parsing on
compositional intents (+46.2pp), generalizes to novel phrasings
(81.8\% OOD accuracy), and produces zero constraint violations in
end-to-end routing. The reachability separation analysis reveals
that apparent performance gaps under polar constraints are topological
artifacts, not routing deficiencies.

Future work will address the semantic verification gap through
constraint satisfiability pre-solving, extend the GNN to accept
constraint features as input, and validate the system on real
operator intent traces from production constellations.



\end{document}